\documentclass[conference, onecolumn]{IEEEtran}
\usepackage{amsmath}
\usepackage{amsfonts}
\usepackage{amsthm}
\usepackage{graphicx}
\usepackage{float}
\usepackage{stfloats}
\usepackage{paralist}
\usepackage{tikz}
\usetikzlibrary{arrows,automata}
\usepackage{cite}
\newtheorem{remark}{Remark}
\newtheorem{theorem}{Theorem}

\newcommand{\pip}{\frac{\pi_0}{\pi_1}}
\allowdisplaybreaks
\begin{document}
\title{Adaptive Molecule Transmission Rate for Diffusion Based Molecular Communication}
\author{\IEEEauthorblockN{Mohammad Movahednasab\IEEEauthorrefmark{1},
Mehdi Soleimanifar\IEEEauthorrefmark{1}, Amin Gohari\IEEEauthorrefmark{1},\\
Masoumeh Nasiri Kenari\IEEEauthorrefmark{1} and Urbashi Mitra\IEEEauthorrefmark{2}}
\IEEEauthorblockA{\IEEEauthorrefmark{1} Sharif University of Technology \IEEEauthorrefmark{2}University of Southern California (USC)}}
\maketitle
\begin{abstract}
In this paper, a simple memory limited transmitter for molecular communication is proposed, in which information is encoded in the diffusion rate of the molecules. Taking advantage of memory, the proposed transmitter reduces the ISI problem by properly adjusting its diffusion rate. The error probability of the proposed scheme is derived and the result is compared with the lower bound on error probability of the optimum transmitter. It is shown that the performance of introduced transmitter is near optimal (under certain simplifications). Simplicity is the key feature of the presented communication system: the transmitter follows a simple rule, the receiver is a simple threshold decoder and only one type of molecule is used to convey the information. 
\end{abstract}
\section{Introduction}
\label{sec:intro}
New applications such as smart drug delivery and health monitoring give rise to the importance of molecular communication, a new paradigm for communication between nanomachines over a short (nanoscale or microscale) range. In molecular communication, information is carried by molecules, rather than electrons or electromagnetic waves \cite{Book1, Book2}.
Several types of molecular communication have been considered, among them, diffusion based communication, which corresponds to traditional wireless communication\cite{AAA}, is of great interest, since it does not require any prior communication link infrastructure. In diffusion based communication, the transmitter nanomachine releases information molecules in the environment. These released molecules diffuse randomly until they hit the receiver nanomachine. \cite{Pierboon2010}, \cite{Eckford2007}.
Due to the random nature of molecular propagation, diffusion based communication suffers from inter symbol interference (ISI). Several solutions have been proposed to mitigate ISI (e.g. see \cite{Leeson2012, Atakan2012, Ko2012, Shih2012}). In \cite{AR}, a new modulation technique, named Molecular Concentration Shift Keying (MCSK), is suggested. Exploiting two types of molecules, while MCSK eliminates the interference from the last transmitted symbol and reduces the error probability, it suffers from interference due to earlier transmissions. A solution based on adding intelligence to receiver is suggested in \cite{MO}. where the receiver stores the last decoded bits in memory to make an estimation of current interference level, and uses this estimation to adjust the threshold for decoding the current bit. In \cite{SH}, a linear and time invariant model is presented and the optimal receiver is derived, under this model. However, this receiver is too complex to be implemented in practice. In \cite{fekri}, the authors considered a deterministic noiseless diffusion channel with memory, and proposed using different symbol durations to deal with ISI by taking into account the channel binary concentration state. They then computed the channel capacity by adapting the Shannon telegraph channel method.
In this paper, we propose a simple transmitter which significantly reduces ISI by adaptively adjusting transmission rates to stabilize the rate of molecules at the receiver, enabling the use of a simple fixed threshold receiver. To this end, the transmitter exploits memory to keep a partial transmission history, so that it can estimate the interference rate that will be experienced at the receiver side, in order to determine a proper diffusion rate for the current transmission.
It is shown that the proposed transmitting protocol is near optimal by obtaining a tight lower bound on the error probability and comparing its performance with the lower bound. 
The rest of this paper is organized as follows: the system model is described in Section \ref{sec:sysmod} and the proposed transmitter is introduced in Section \ref{sec:adapttrans}. In Section \ref{sec:performance}, performance bounds are obtained and it is shown that the proposed transmitter is near optimal, and finally, the paper is concluded in Section \ref{sec:conclu}.
Throughout this paper all the logarithms are in base $e$.
\section{System Model}
\label{sec:sysmod}
\subsection{Transmitter and Receiver Model}
We consider the communication model described in \cite{MO}-\cite{AR}. Transmission occurs in equal time slots $T_s$, called the symbol duration. The input bit stream ($b_1,b_2,\cdots$) comprises of i.i.d. $\mbox{Bernoulli}(1/2)$ rv's. At the beginning of each time slot $i$, a number of molecules is released by the transmitter that is approximated as
Poisson variable, i.e., $\mbox{Poisson}(X_i)$, where the rate $X_i$ is determined by the information that the transmitter wishes to transmit. In this work, in contrast to \cite{MO}-\cite{AR}, the transmitter is an intelligent device consisting of $M$ bits of memory, a transmission function, and a memory updating function, as shown in Fig. \ref{transmodel}. The transmitter produces $X_i$ based on the current information symbol $b_i$ and the $M$ bits that are stored in the memory.
\begin{figure}
\centering
\includegraphics[width=3.5in]{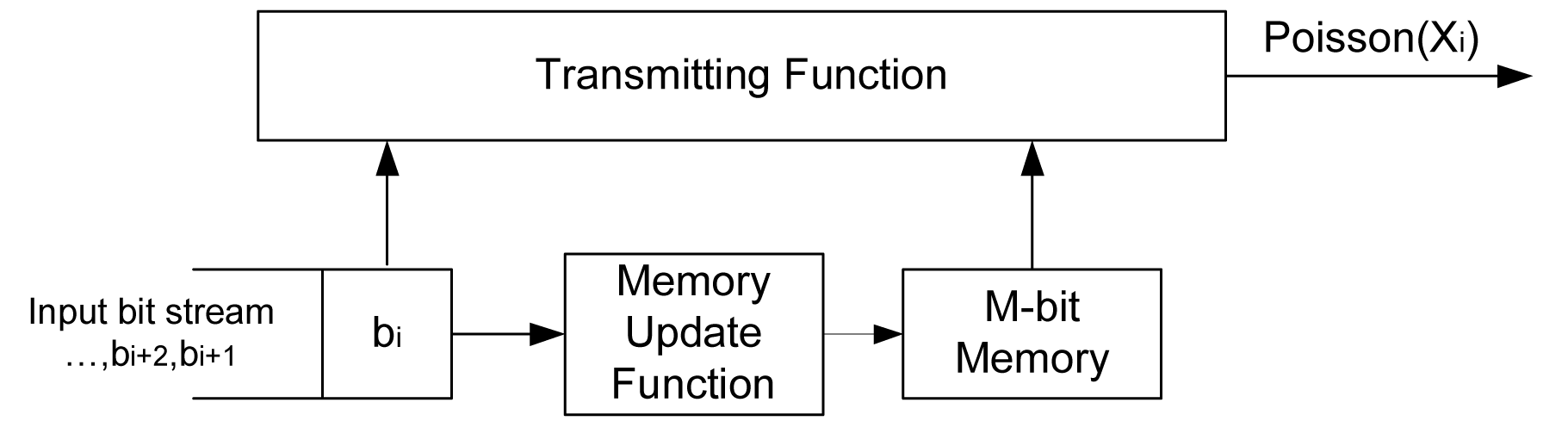}
\caption{Transmitter Model}
\label{transmodel}
\end{figure}
The receiver is as simple as possible, it is just a threshold decoder with a fixed threshold. So, all the intelligence is kept in the transmitter to make the receiver simple.
\subsection{Channel Model}
We consider a diffusion based model for communication. That is, molecules are freely released in the fluid where they propagate via Brownian motion. Although highly random, Brownian motion is always available, and has the advantage of zero energy propagation cost \cite{Book2}. The released molecules continuously diffuse until they hit the receiver, once they reach the receiver they will be absorbed and removed from the media. Let us denote by $\Pi = [ \pi_0,\pi_1,...] $ the sequence of hitting probabilities in consecutive time slots, i.e., $\pi_i$ is the probability that a released molecule at the beginning of the $k$-th time slot arrives at the receiver during the $k+i$-th time slot, $i=0,1,2,\cdots$. In the simple case of one dimensional Brownian motion in a uniform medium, it is shown that the first arrival time follows an Inverse Gaussian distribution\cite{Book1} and the vector $\Pi$ can be obtained easily (related equations can be found in\cite{MO}). There may exist other uninvited sources of molecules of the same type as our information molecules, which are treated as noise, and is modeled by a Poisson random variable with parameter $\lambda_0$. Considering the channel characteristics just described, the channel output in each time slot is influenced by three factors: \begin{inparaenum}[(i)]
\item Input rate in the current time slot transmission ($X_i$), 
\item Input rates of previous transmissions ($X_{i-1},X_{i-2},...$),
\item The noise parameter $\lambda_0$.
\end{inparaenum}
Using the thinning property of Poisson distribution and the fact that sum of independent Poisson rv's is itself a Poisson, the channel output is a Poisson rv, and is found in \cite{AR} as follows:
\begin{align} 
\begin{split}
Y_i & \sim \mathsf{Poisson}\left(\pi_0 X_i +\sum_{k=1}^{\infty}\pi_k X_{i-k}+\lambda_0\right)
\\&= \mathsf{Poisson}(\pi_0 X_i +I_i+\lambda_0),
\label{channeloutput}
\end{split}
\end{align}
where $I_i$ denotes the interference term at $i$th time slot. We say that the channel has memory $M_c$ if \begin{align}\pi_k=0,\qquad \forall k>M_c.\label{eqn:memorychannel}\end{align}
\section{The ISI problem and Adaptive Diffusion Rate}
\label{sec:adapttrans}
Consider a conventional OOK (On-Off Keying) modulator, where the transmitter sends a constant rate of molecules for bit `1' and nothing for bit `0'. The performance of this modulator (which is commonly used in most positive systems like optical communications) is degraded in diffusion based channels due to the ISI effect.\footnote{A positive system is a channel that only accepts non-negative inputs.} This is due to the fact that the rate of received molecules at the destination depends not only on the current transmission, but also on previously transmitted bits (or equivalently rates). Fig. 2 shows a sample of the received rate of molecules at the destination. One can realize from this figure that a constant threshold decoder is not an optimal or near optimal solution for this channel. Even though we are using only two diffusion rates at the transmitter for bits `0' and `1', the absorption rates at the receiver are not fixed. 
We now introduce our simple modulation: we use a constant threshold decoder at the receiver, but to avoid the ISI problem mentioned above, we modify the OOK strategy by adapting the current transmission rate to the expected value of interference at the receiver, i.e., for transmitting bit `1', we diffuse less molecules if we expect a high concentration of molecules at the receiver due to previous transmissions; for bit `0', we send nothing. 
\begin{figure}[!t]
\centering
\includegraphics[width=3.5in]{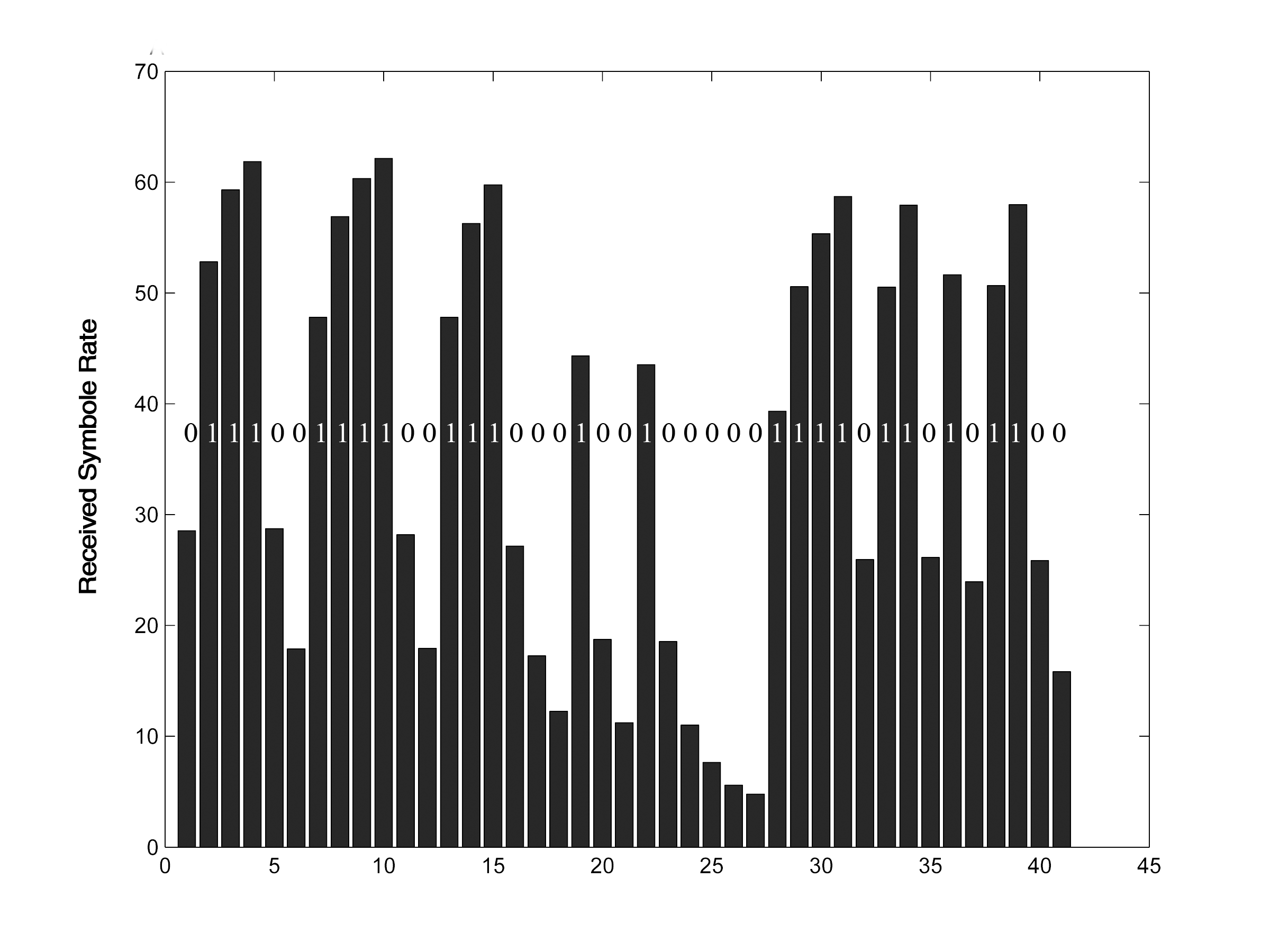}
\caption{Sample Received Mean Number of Molecules}
\end{figure}
To achieve the expected absorption rate of molecules at the receiver, the transmitter needs infinite bits of memory to remember all previously transmitted symbols which affect the current transmission, in order to estimate the interference level at the receiver side, and hereby determine the rate of molecules that should be released in the current time slot for bit one. To make the suggested scheme practical, we need to adapt it to a limited memory system, in which only a few bits of memory is available. As a result, the transmitter cannot estimate the exact value of the interference $I_i$ in Eq. \eqref{channeloutput}; however, we will see that limited memory systems reach a performance close to that of unlimited memory. 
Suppose a transmitter with $M$ bits of memory is available. Our proposed encoder uses the memory to store the last $M$ transmitted information bits, which indicates the transmitter's state. A typical transmission function can be considered as follows: a transmitting rate is assigned to each $2^{M}$ states of memory, such that if the current input bit is `1', the rate corresponding to the state is selected for the transmission, and if the current input bit is `0', nothing is transmitted. The proposed transmitter can be modeled as finite state machine with $2^M$ states. Fig. 3 shows the state diagram for a transmitter with 2 bits memory, in which $L_{ij}\;i,j\in\{0,1\}$ denotes the rate of molecules released, if the input bit is `1' and the transmitter state is $S_{ij}\;i,j\in\{0,1\}$, noting that nothing is transmitted for bit `0'. The subindexes $i,j$ in the above notation are equal to two last transmitted information bits.
\begin{figure}
\centering
\begin{tikzpicture}[->,>=stealth',shorten >=1pt,auto,node distance=3.1cm,
thick,main node/.style={circle,draw,font=\sffamily\Large\bfseries}]
\node[main node] (1) {$S_{00}$};
\node[main node] (2) [below of=1] {$S_{01}$};
\node[main node] (3) [right of=1] {$S_{11}$};
\node[main node] (4) [right of=2] {$S_{10}$};
\path[every node/.style={font=\sffamily\small}]
(1) edge node[left,near end] {$1/L_{00}$} (2)
edge [loop above] node [near end]{$0/0$} (1)
(2) edge node [very near end]{$1/ L_{01}$} (3)
edge[bend right] node [very near end]{$0/0$} (4)
(3) edge [loop above] node [near end]{$1/ L_{11}$} (3)
edge node[near end]{$0/0$} (4)
(4) edge [bend right] node [very near end] {$1/ L_{10}$} (2)
edge node[right,very near end] {$0/0$} (1);
\end{tikzpicture}
\caption{Transmitter State Diagram}
\end{figure}
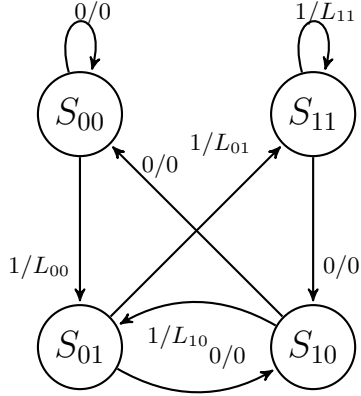
The only remaining step is to define levels ($L_{ij}$) such that for each bit `1', we get close to a constant rate at the receiver. The exact value of previous diffusion rates is not known, but the knowledge of last $M$ transmitted bits gives us partial information about previously transmitted rates. To clarify the concept, suppose only two bits of memory are available ($M=2$), and the last two transmitted bits are `1' which means we are in state $S_{11}$, then we are sure that the last transmitted rate is either $L_{01}$ or $L_{11}$ each with probability $1/2$, depending on whether the third earlier bit is zero or one respectively. We can consider the expected value, i.e., $(L_{01}+L_{11})/2$, for the last transmitted rate. For each memory state $S_{ij} \;\;i,j\in\{0,1\}$, we can write down the expected value of the rate of molecules that will be received at the destination conditioned on being on that state, and by setting all expected rates equal to a preselected constant, we get a system of linear equations. By solving these equations simultaneously, transmitting rates corresponding to different states are obtained, and hereby the transmission function is determined. 
The set of equations for a transmitter with two bits memory that stores bits $(B_{i-1}=b_{i-1}, B_{i-2}=b_{i-2})$, assuming the channel has memory equal to $M_c$ (see Eq. \ref{eqn:memorychannel}), is derived as follows: when $B_i=1$, we transmit $$L_{b_{i-1}, b_{i-2}}=\frac{C-\mathbb{E}[I_i|B_{i-1}=b_{i-1}, B_{i-2}=b_{i-2}]}{\pi_0},$$ where $C$ is a constant (preselected rate that is expected to be achieved at the destination for each information bit `1') and $\mathbb{E}[I_i|B_{i-1}=b_{i-1}, B_{i-2}=b_{i-2}]$ is the conditional expected value of interference. Using the fact that $X_{i-k}$ is a function of $(B_{i-k}, B_{i-k-1}, B_{i-k-2})$, we can compute this conditional expected value as follows:\small
\begin{align*}&\mathbb{E}[I_i|B_{i-1}=b_{i-1}, B_{i-2}=b_{i-2}]
\\&=\sum_{k=1}^{\infty}\pi_k \mathbb{E}[X_{i-k}|B_{i-1}=b_{i-1}, B_{i-2}=b_{i-2}]
\\&=\pi_1\mathbb{E}[X_{i-1}|B_{i-1}=b_{i-1}, B_{i-2}=b_{i-2}]
\\&\quad+\pi_2\mathbb{E}[X_{i-2}|B_{i-2}=b_{i-2}]
+\sum_{k=3}^{\infty}\pi_k \mathbb{E}[X_{i-k}]
\\&=\pi_1L_{b_{i-1},b_{i-2}}/2
+\pi_2(L_{b_{i-2},0}+L_{b_{i-2},1})/4
+\sum_{k=3}^{\infty}\pi_k \mathbb{E}[X_{i-k}].
\end{align*} 
\normalsize
Simulation results (Fig. \ref{comp}) show that a limited memory transmitter with only two bits of memory can reach the performance of the unlimited memory transmitter, and the performance of our simple system is comparable to the system with memory at the receiver side introduced in \cite{MO}, which is more complex. In Section \ref{sec:performance}, we show that this transmission protocol is near optimal under certain assumptions.
\section{Performance Bound and Numerical Results}
\label{sec:performance}
In this section, we are interested in evaluating the performance of proposed transmission protocol. To understand how well the proposed scheme performs, we need to compare its bit error probability with that of the transmitter with an optimal transmission function, i.e., transmitter with minimum error probability. Unfortunately, deriving the optimal transmission function by minimizing the error probability is not a simple task, particularly when we are dealing with fixed values of memory, since the decision on the transmitting rate in each time slot affects not only the current time slot, but also all following time slots, which makes the problem highly complicated. As a result, we provide a combination of partial results, insights and simulation to make a case for this modulation scheme. 
We first make an observation in Subsection IV. A where we show that under a simplifying assumption the optimal transmission function matches the proposed one at all points in the real line except for possibly one point. In Subsection IV. B we derive a lower bound on the error probability for the special case of channels with one symbol memory, and we will see that our system performance is close enough to the lower bound. Further, it is shown that the optimum transmission function must send zero when input $X_i$ is zero, thus confirming our choice for this. Numerical results are given in Subsection IV. C.
\subsection{Fixed Interference distribution} 
For now, assume that the transmitter has unlimited memory, so it knows the exact value of the interference rate, $I_i$, experienced at the receiver in each time slot. Let us denote the transmission function by $f:\mathbb{R}_{+}\mapsto\mathbb{R}_{+}$ which maps the interference rate to a transmitting rate, such that the bit error probability is minimized. That is for $b_i=1$, we select $X_i=f(I_i)$, where $I_i$ is the current level of the interference at the receiver side, and for $b_i=0$ we transmit nothing, i.e. $X_i=0$. It is necessary to assume a constraint on the average transmission rate for bit one, i.e., if we use the channel $n$ times we should have $\frac{1}{n}\sum_{i=1}^nX_i\leq K$ for some constant $K$; this should be interpreted as an input power constraint. 
The distribution of interference itself is determined by our choice of $f$. However, let us fix some distribution $q_I$ on $I$ and ask for the best function for that distribution. This can give us some insights about $f$ in general. That is given $q_I$ we are interested in a function $f$ such that $\mathbb{E}_q[X]=\mathbb{E}_q[f(I)]= K$. From Eq. (\ref{channeloutput}) for a threshold decoder with threshold equal to $T$, the error probability is equal to
\begin{align}
\begin{split}
P_e&=\frac{1}{2}P_{e|0}+\frac{1}{2}P_{e|1}\\
&=\mathbb{E}_I\Bigg[\frac{1}{2}\sum_{y>T}e^{-I}\frac{I^y}{y!}+\frac{1}{2}\sum_{y<T}e^{-I-\pi_0f(I)}\frac{(I+\pi_0f(I))^y}{y!} \Bigg ].
\label{errorprob}
\end{split}
\end{align}
We would like to select $f(\cdot)$ in a way to minimize Eq. (\ref{errorprob}). Alternatively, we would like to minimize
\begin{align} 
\begin{split}
\mathbb{E}_I\left[e^{-I-\pi_0f(I)}\sum_{0\leq y<T}\frac{(I+\pi_0f(I))^y}{y!}\right],
\end{split}
\label{fixedDistMin}
\end{align}
subject to $\mathbb{E}[f(I)]= K$.
Using the Lagrange multiplier technique we show in Appendix \ref{fixeddist} that there exists some $C>0$ such that the optimum $f(\cdot)$ is equal to 
$$f(i)=\frac{C-i}{\pi_0},$$
for all $i\geq 0$, except for possibly one particular $i^*$. Thus, at the transmitter $\pi_0 f(i) + i$ is kept constant for almost all possible values of interference $I=i$, matching our proposed scheme. This statement holds universally for any arbitrary fixed distribution on interference. 
Although we have simplified the original problem, but the results obtained in this section provides insights about the solution of the original problem.
\subsection{Lower Bound on Error Probability}
In this section, we derive a lower bound on the error probability. To track the dependence of the distribution of the interference in terms of the transmission function, we restrict the channel to be one with only one symbol memory. i.e., $M_C=1$ and $\pi_i=0$ for all $i\geq 2$. In this case, at time slot $j$, interference value $I_j$ is equal to $\pi_1 X_{j-1}$. Here, we also assume that the released molecules have a higher probability of arriving at the receiver in their transmission time slot compared to the next time slot, i.e., $\pi_0>\pi_1\geq \pi_2=\pi_3=\cdots=0$.
The transmitter works as before: if the input bit is `1', the transmission function determines the rate of released molecules, and if the input bit is `0', nothing is transmitted. In this section, the transmitter is assumed to know the exact value of the interference. 
\begin{remark}
One might consider releasing a constant rate of molecules for bit `0' instead of nothing. We prove in Appendix \ref{remark} that, in an optimal transmission function, nothing should be transmitted for bit `0'.
\end{remark}
Since the channel is assumed to have memory one, any `0' input would clear the channel memory. Taking this point into consideration, we model the transmitter as an infinite state machine, with each state $S_i^{'},\;i\in\{0,1,..\}$, the subscript $i$ describes how many `1's were transmitted in sequence after a memory reset via a `0' bit., i.e. the last $i+1$ bits at the current time slot $j$ are $(b_{j-i},\cdots, b_j)=(0,1,1,\cdots,1)$. The transmitter state diagram is shown in Fig. \ref{tsd}.\\
\begin{figure}\begin{center}
\begin{tikzpicture}[>=stealth',shorten >=1pt,auto,node distance=2cm]
\node[state] (S) {$S_0^{'}$};
\node[state] (q1) [right of=S] {$S_1^{'}$};
\node[state] (q2) [right of=q1] {$S_2^{'}$};
\node[state] (q3) [right of=q2] {$S_3^{'}$};
\node (q4) [right of=q3] {$...$};
\path[->]
(S) edge [loop above] node {$\frac{1}{2}$} (S)
edge node {$\frac{1}{2}$} (q1)
(q1) edge [bend left] node {$\frac{1}{2}$} (S)
edge node {$\frac{1}{2}$} (q2)
(q2) edge node {$\frac{1}{2}$} (q3)
edge [bend left] node {$\frac{1}{2}$} (S)
(q3) edge node {$\frac{1}{2}$} (q4)
edge [bend left] node {$\frac{1}{2}$} (S);
\end{tikzpicture}\end{center}
\caption{Transmitter State Diagram}
\label{tsd}
\end{figure}
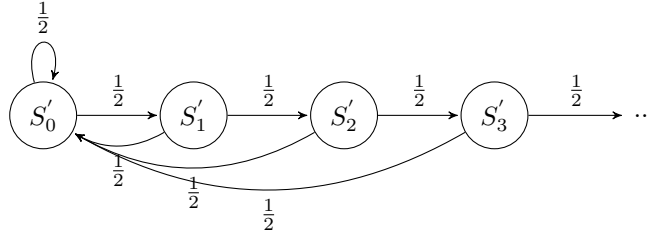
Let us denote by $a_i$ the interference rate at the receiver side, when the transmitter is in state $S_i^{'}$. The choice of the transmission function determines the value of $a_i,\;i\in{1,2,..}$. More precisely, assume that in time slot $j$, we are in state $S_i^{'}$ and the interference value is $I_j=a_i$. If $b_j=0$, we send nothing and the state is reset to $S_0^{'}$. We send nothing in this slot and the channel memory, $I_{j+1}=\pi_1X_{j}=0$. Thus, $a_0=0$. However if $b_{j}=1$, we transmit $X_j=f(a_i)$ and move from state $S_i^{'}$ to $S_{i+1}^{'}$. The interference value at state $S_{i+1}^{'}$ will be then equal to $I_{j+1}=\pi_1X_{j}=\pi_1 f(a_i)$. Then,
$$a_{i+1}=\pi_1 f(a_i).$$
From Fig. \ref{tsd}, it can be easily shown that the probability of being in each state $S_i^{'}$ equals $P_i=\big(1/2\big)^{i+1}$ for $i\in\{0,1,..\}$. To determine the optimum transmission function, we must first compute the bit error probability, which, conditioned on the value of the state, is equal to
\begin{align} 
\begin{split}
\mathbb{P}_e=& \dfrac{1}{2}P_{e|1}+\dfrac{1}{2}P_{e|0}\\
=&\dfrac{1}{2}\sum_{i=0}^{\infty}P_i\left[e^{-(a_i+\dfrac{\pi_0}{\pi_1}a_{i+1})}\sum_{ y\leq T}\dfrac{(a_i+\dfrac{\pi_0}{\pi_1}a_{i+1})^y}{y!}\right]\\
&+\dfrac{1}{2}\sum_{i=0}^{\infty}P_i\left[e^{-a_i}\sum_{y>T}\frac{a_i^y}{y!}\right],
\label{errorprob2}
\end{split}
\end{align}
where $P_i=\big(1/2\big)^{i+1}$ was the probability of being in state $S_i^{'}$, and $\pip a_{i+1}=\pi_0f(a_i)$ equals the average received molecules due to the current transmitted bit `1', when we are in state $S_i^{'}$. 
We are interested in minimizing Eq. (\ref{errorprob2}) over all $a_i\geq0$ subject to $$
\mathbb{E}[X]=\sum_{i=0}^{\infty}P_if(a_i)=\sum_{i=0}^{\infty}P_ia_{i+1}/\pi_1= K,$$ where $K$ is the power constraint defined in Eq. (\ref{errorprob}). Since the constraint is linear, the Karush-Kuhn-Tucker regularity conditions hold and can be written down as
\begin{align}
\begin{split}
&P_ie^{-a_i-\dfrac{\pi_0}{\pi_1}a_{i+1}}\dfrac{\left(a_i+\dfrac{\pi_0}{\pi_1}a_{i+1}\right)^T}{T!} \\
&+P_{i-1}\dfrac{\pi_0}{\pi_1}e^{-a_{i-1}-\dfrac{\pi_0}{\pi_1}a_{i}}\dfrac{\left(a_{i-1}+\dfrac{\pi_0}{\pi_1}a_i\right)^T}{T!} \\
&- P_{i}e^{-a_i}\dfrac{a_i^T}{T!}=-\mu_i+\lambda P_i \; \;\; \text{for}\; i= 1,2,... 
\label{eq:KKT}
\end{split}
\end{align}
Finding the $a_i$'s satisfying these equations is equivalent to finding the optimum transmission function, and as a result, the minimum error probability that can be achieved. A closed form solution for these equations does not exist. Instead of finding the exact value of the minimum error probability, we derive a lower bound. To do so, we consider only the dominant terms contributing to the error probability which are the terms containing $a_1$ and $a_2$. Let us define $\mathsf{P}(a_1,a_2)$ as below:\\
\begin{align}
\begin{split}
\mathsf{P}(a_1,a_2):=&\frac{1}{2}P_0 e^{-\pip a_1}\sum _{y=0}^{T}\frac{\left(\pip a_1\right)^y}{y!}+\\
&\frac{1}{2}P_1 e^{-a_1-\pip a_2}\sum _{y=0}^{T}\frac{\left(a_1+\pip a_2\right)^y}{y!}+ \\
&\frac{1}{2}P_1 e^{-a_1}\sum_{y=T+1}^{\infty}\frac{a_1^y}{y!}+\frac{1}{2}P_2e^{-a_2}\sum_{y=T+1}^{\infty}\frac{a_2^y}{y!}.
\label{error_reduced}
\end{split}
\end{align}
From Eq. (\ref{errorprob2}) and Eq. (\ref{error_reduced}) it is clear that (since $a_1,...,a_n\geq0$)
$\mathsf{P}(a_1,a_2)\leq\mathbb{P}_e$. It can be easily observed that 
\begin{align*}\min_{a_1,a_2\geq 0} \mathsf{P}(a_1,a_2)\leq\min_{a_j\geq0 \, ,\sum P_j a_{j+1}\leq \pi_1 k} \mathbb{P}_e.\end{align*}
{Minimizing $\mathsf{P}$ over non-negative values of $a_1$ and $a_2$ requires solving $\nabla \mathsf{P}(a_1,a_2)=0$, which results in equations of the form of Eq. (\ref{eq:KKT}) and a closed form solution for them does not exit. However, we can prove some properties for the solutions. Using these properties, we develope a simple tight lower bound on $\mathsf{P}(a_1,a_2)$ which is also a lower bound on the minimum error probability. In the rest of this section, we first present some properties for the solutions of $\nabla \mathsf{P}(a_1,a_2)=0$ in Theorem 1 and then the lower bound is demonstrated in Theorem 2.}
\begin{theorem}\label{th1} Given a fixed threshold $T$, let $\theta$ be the unique solution of the following equation:
\begin{align}\frac{\log\left(2\theta^{T+1}\right)}{\theta-1}=T.\label{eqn:def-theta}\end{align}
Then for all $\pip$ greater than $\theta$ and for all solutions of $\nabla \mathsf{P}(a_1,a_2)=0$, we have that $a_2$ is less than $T$. Moreover all of the solutions of $\nabla \mathsf{P}(a_1,a_2)=0$ in the interval $0< a_1, a_2<T$ satisfy $$\frac{\log\left(2(\pip)^{T+1}\right)}{\pip-1}<a_1.$$
Also, the function $\mathsf{P}(a_1,a_2)$ does not have any local minimum at the boundary points $a_1=0$ or $a_2=0$. 
\end{theorem}
The proof is given in Appendix \ref{th1proof}. Using the bound on $a_1$ given in Theorem 1, we can find a lower bound on the minimum error probability, which is given in Theorem 2.
\begin{theorem}\label{th2} { Given a threshold greater than four molecules, i.e. $T\geq 4$, we have} $$\mathsf{P}(a_1,a_2)>\mathbb{L}\left(\frac{\log\left(2(\pip)^{T+1}\right)}{\pip-1},\frac{T}{\pip+1}\right),$$ for any $a_1,a_2>0$ and any $\pip>\theta$, where $\theta$ was given in Eq. \eqref{eqn:def-theta}. The function $\mathbb{L}$ is defined as follows:
\begin{equation}
\begin{split}
\mathbb{L}(a_1,a_2):=&e^{-a_1}\frac{a_1^T}{T!}\left(\frac{1}{8}\frac{a_1}{T+1}+\frac{1}{8 \pip}\right)+\\
&e^{-a_2}\frac{a_2^T}{T!}\left(\frac{1}{16}\frac{a_2}{T+1}+\frac{1}{16 \pip}-\frac{1}{16(\pip)^2}\right).
\label{lower_bound}
\end{split}
\end{equation}
\end{theorem}
The proof is given in Appendix \ref{th2proof}. Numerical results show that the lower bound given in Eq. (\ref{lower_bound}) is quite tight and the error probability of our simple proposed transmitter is very close to the bound for typical values of channel parameters. 
\subsection{Numerical Results}
In this section, we evaluate the performance of our proposed scheme. For all results presented in this section, the molecule hitting probabilities are calculated as in \cite{MO}.
In Fig. \ref{comp}, the performance of the proposed transmitter is depicted versus the symbol rate for different number of transmitter bit memory. For the comparisons, the error probability plots for the transmitter with infinite memory (known interference), the conventional transmitter with simple threshold decoder and with the decoders introduced in \cite{MO} are included. As it can be seen, with only two bits of memory in the transmitter, we can reach the performance of an unlimited memory system, which is aware of the entire transmission history. From this figure, if the ISI is totally neglected, the error probability is high and the system is unreliable. The system with memory at the receiver \cite{MO}, outperforms our system, but the current proposed system is superior in the sense of complexity.
\begin{figure}
\centering
\includegraphics[width=3.7 in]{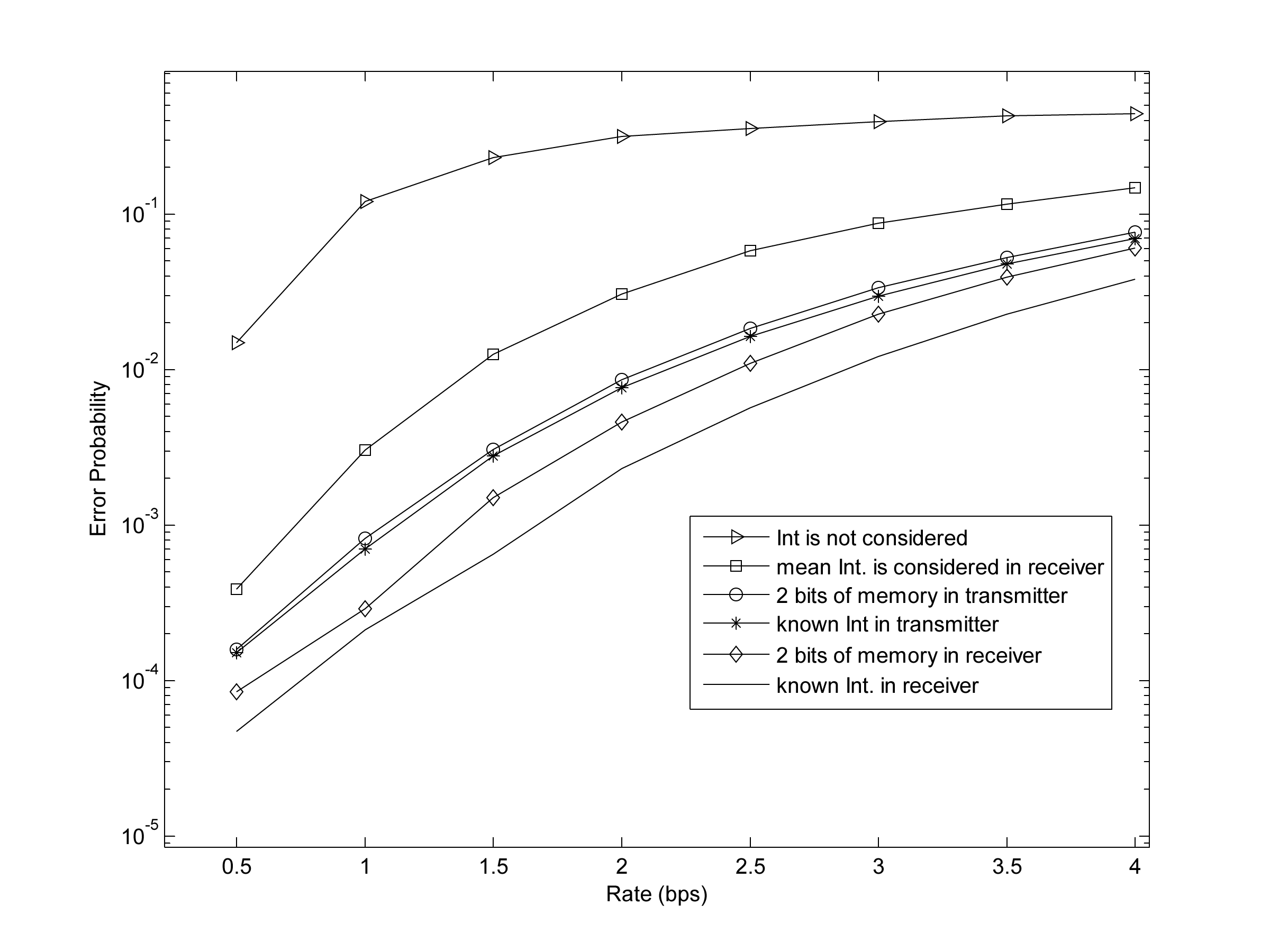}
\caption{ Proposed System Performance. Parameters for this figure: average transmitted molecules per bit = 80, noise rate $\lambda_0=10$, channel memory = 10.}
\label{comp}
\end{figure}
Fig. \ref{lb} shows the derived lower bound versus the transmission symbol rate. For comparison, the error probability of the proposed system and the minimum error probability obtained by minimizing Eq. (\ref{errorprob2}) numerically are included as well. As can be seen, the proposed system is near optimal and the lower bound is tight.
\begin{figure}
\centering
\includegraphics[width=3.7in]{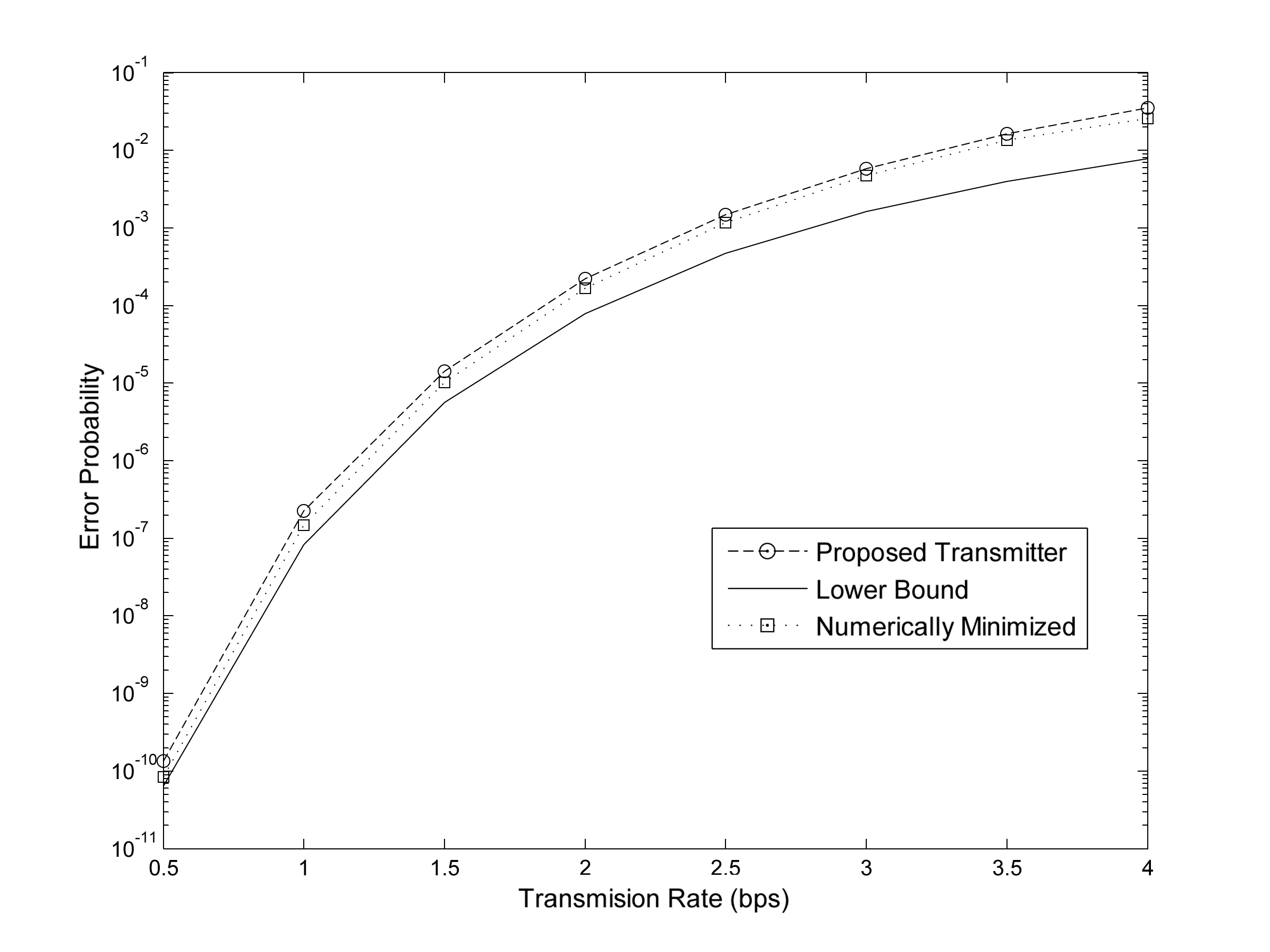}
\caption{Comparing simple transmitter performance with lower bound and minimum error probability. Parameters for this figure: average transmitted molecules per bit = 80, noise rate $\lambda_0$=0, channel memory = 1.}
\label{lb}
\end{figure}
In Fig. \ref{dist-error}, the performance of proposed transmitter is compared with the conventional OOK transmitter for different distances between transmitter and receiver. From this Fig. as the distance increases, the channel memory also increases, and as a result the performance degrades. 
\begin{figure}
\centering
\includegraphics[width=3.7in]{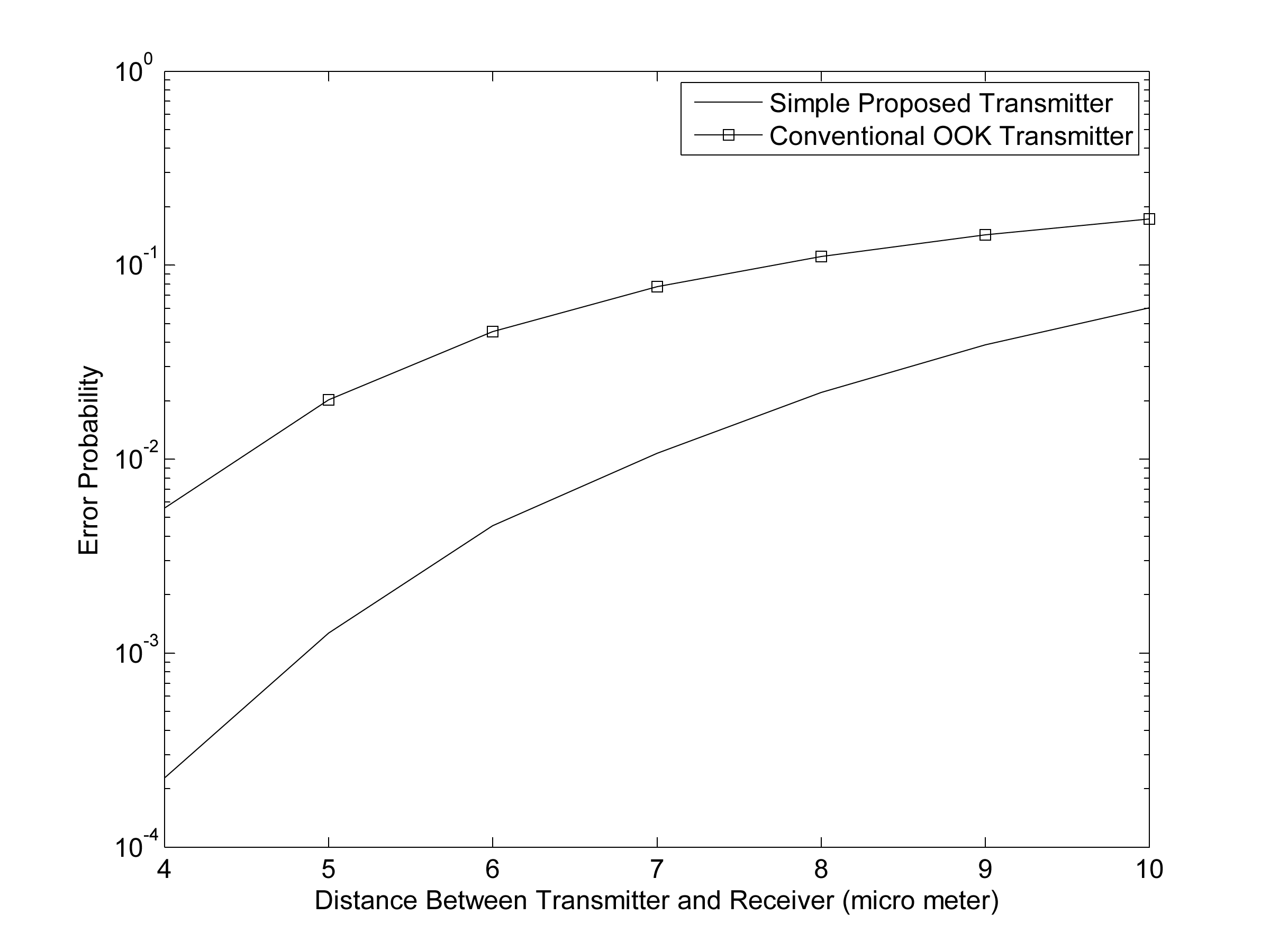}
\caption{Comparing simple transmitter performance with conventional OOK. Parameters for this figure: average transmitted molecules per bit = 80, noise rate $\lambda_0=0$, channel memory = 10.}
\label{dist-error}
\end{figure}
\section{Conclution}
\label{sec:conclu}
In this paper, we focused on the ISI problem in diffusion based molecular communication. We proposed a simple system consisting of a transmitter with varying molecular transmission rate and a simple threshold decoder. We observed that using a limited memory system we can reach a performance close to that of unlimited memory. Also we showed that under certain simplifications our transmitter is near optimal.

\appendices
\section{Proof of Theorem \ref{th1} \label{th1proof}}
Considering Eq. \eqref{error_reduced}, $\nabla \mathsf{P}(a_1,a_2)=0$ leads to:
\begin{align}
\begin{split}
2\pip e^{-\pip a_1}&\frac{\left(\pip a_1\right)^T}{T!}+
e^{-a_1-\pip a_2}\frac{\left(a_1+\pip a_2\right)^T}{T!}=e^{-a_1}\frac{a_1^T}{T!}
\label{eq:1}
\end{split}
\vspace{-0.5cm}
\end{align}
\begin{align}
\begin{split}
&2\pip e^{-a_1-\pip a_2}\frac{\left(a_1+\pip a_2\right)^T}{T!}=e^{-a_2}\frac{a_2^T}{T!}.
\label{eq:2}
\end{split}
\end{align}
\emph{Proof of the inequality on $a_1$ when $a_1<T$:} Assume that $\nabla \mathsf{P}(a_1,a_2)=0$ and $a_1<T$. From Eq. \eqref{eq:1} we have that
\begin{align}2\pip e^{-\pip a_1}\frac{\left(\pip a_1\right)^T}{T!}< e^{-a_1}\frac{a_1^T}{T!}.\label{eqn:newwqau1}\end{align}
The above equation can be simplified as follows:
\begin{align}2\pip e^{(1-\pip) a_1}{\left(\pip \right)^T}< 1.\label{eqn:newwqau2}\end{align}
which implies that $$\frac{\log\left(2(\pip)^{T+1}\right)}{\pip-1}<a_1.$$ 
\emph{Proof of the inequality $a_2<T$:} Assume that $\nabla \mathsf{P}(a_1,a_2)=0$. We would like to prove that $a_2<T$.
We prove this by contradiction. If $a_2\geq T$, then $a_1+\pip a_2\geq T$. Since $e^{-x}\frac{x^T}{T!}$ is a decreasing function for $x\geq T$, from Eq. \eqref{eq:2}, we have that
$$2\pip e^{-\pip a_2}\frac{\left(\pip a_2\right)^T}{T!}\geq 2\pip e^{-a_1-\pip a_2}\frac{\left(a_1+\pip a_2\right)^T}{T!}=e^{-a_2}\frac{a_2^T}{T!}.
$$
Hence,
$$2{\left(\pip \right)^{T+1}}\geq e^{a_2(\pip-1)}.$$
Thus,
$$\frac{\log\left(2(\pip)^{T+1}\right)}{\pip-1}\geq a_2.$$ 
We will arrive at a contradiction if $$\frac{\log\left(2(\pip)^{T+1}\right)}{\pip-1}<T.$$
The left hand side is a decreasing function in terms of $\pip$. Furthermore, at $\pip=\theta$ given in Eq. \eqref{eqn:def-theta} we have equality. This proves both the uniqueness of the solution to Eq. \eqref{eqn:def-theta} and the desired result for $\pip>\theta$.
\emph{Proof for the boundary cases:} To show that a minimum cannot occur at the boundary, i.e., $a_2=0$ or $a_1=0$, let us consider the gradient. For $a_2=0$
\begin{align}\frac{\partial \mathsf{P}}{\partial a_2}=-\frac{P_1}{2}\pip e^{-a_1}\frac{(a_1)^T}{T!}<0,\end{align}
and for $a_1=0$
\begin{align}\frac{\partial \mathsf{P}}{\partial a_1}=-\frac{P_1}{2} e^{-\pip a_2}\frac{(\pip a_2)^T}{T!}<0.\end{align}

\section{Proof of Theorem \ref{th2} \label{th2proof}}
\begin{proof}
Theorem 1 shows that if $a_2$ is an extremum point for $\mathsf{P}(a_1,a_2)$, it must be in $[0,T]$ for $\pip>\theta$. Furthermore, at $a_2=0$ the minimum does not occur. Therefore, we need to consider the cases of finite $a_2\leq T$ and the boundary case of $a_2\rightarrow\infty$. Here we consider three cases:
\begin{enumerate}
\item[\emph{Case 1:}] The function $\mathsf{P}(a_1,a_2)$ has a global minimizer with $a_1\in [0,T]$ and $a_2\leq T$; at this minimizer $\mathbb{L}\left(\frac{\log(2(\pip)^{T+1})}{\pip-1},\frac{T}{\pip+1}\right)$ is a lower bound on minimum error probability;
\item[\emph{Case 2:}] The function $\mathsf{P}(a_1,a_2)$ has a global minimizer with $a_1>T$ and $a_2\leq T$; then $\mathsf{P}(a_1,a_2)$ is again bounded from below by the above lower bound, assuming that $T>4$;
\item[\emph{Case 3:}] The minimum of $\mathsf{P}(a_1,a_2)$ occurs when $a_2$ converges to infinity. Here again, the lower bound holds, assuming that $T>2$.
\end{enumerate}
\vspace{0.5cm}
\emph{Proof of case 1:} Each summation on $y$ in $\mathsf{P}$, in Eq. (\ref{error_reduced}), is over positive terms. So one can consider only one of the terms in summation instead of all the terms to get a lower bound. Thus, by considering the terms $y=T$ in the first two summations and $y=T+1$ in the third and fourth summations, we have
\begin{align}
&\frac{1}{2}P_0 e^{-\pip a_1}\frac{(\pip a_1)^T}{T!}+\frac{1}{2}P_1 e^{-a_1-\pip a_2}\frac{(a_1+\pip a_2)^T}{T!}+\frac{1}{2}P_1 e^{-a_1}\frac{a_1^{T+1}}{(T+1)!}+\frac{1}{2}P_2 e^{-a_2}\frac{a_2^{T+1}}{(T+1)!}\leq \mathsf{P}(a_1,a_2).
\label{appc-1}
\end{align}
Using Eq. (\ref{eq:1}), Eq. (\ref{eq:2}) and substituting the values of $P_i=\big(\frac{1}{2}\big)^{i+1}$ in Eq. \eqref{appc-1} we get:
\begin{align}
\begin{split}
&e^{-a_1}\frac{a_1^T}{T!}\left(\frac{1}{8}\frac{a_1}{T+1}+\frac{1}{8 \pip}\right)+\\
&e^{-a_2}\frac{a_2^T}{T!}\left(\frac{1}{16}\frac{a_2}{T+1}+\frac{1}{16 \pip}-\frac{1}{16(\pip)^2}\right)\leq \mathsf{P}(a_1,a_2).
\label{2-1}
\end{split}
\end{align}
For $(a_1,a_2)\in [0,T]^2$, the LHS of Eq. (\ref{2-1}) is the sum of two increasing functions of $a_1$ and $a_2$ respectively. So if we find lower bounds on $a_1$ and $a_2$, and substitute in Eq. (\ref{2-1}), we get a lower bound on minimum error probability. To that end, from Theorem \ref{th1}, we know that $a_1$ is bounded from below by $\log(2(\pip)^{T+1}/(\pip-1)$. From Eq. (\ref{eq:2}) a lower bound on $a_2$ can be obtained as follows: using the fact that $2\pip>1$ in Eq. (\ref{eq:2}) implies that
\begin{align}e^{-a_1-\pip a_2}\frac{(a_1+\pip a_2)^T}{T!}< e^{-a_2}\frac{a_2^T}{T!}.\end{align}
Due to the ``increasing/decreasing property" of $e^{-x}x^T/T!$ for $x<T$ and $x>T$, we get that $a_1+\pip a_2>T$ as $a_2\leq T$. Furthermore, since $g(x)=e^{-x}x^T/T!$ has the property that $g(T+\alpha)\geq g(T-\alpha)$ for all $\alpha\in[0,T]$ we get that $a_1+\pip a_2-T>T-a_2$, or 
\begin{align}a_1+(\pip+1) a_2>2T,\end{align}
Using the fact that $a_1\leq T$, we get the following bound on $a_2$:
\begin{align}T\geq a_2>\frac{2T-a_1}{\pip+1}\geq\frac{T}{\pip+1}.\end{align}
Let us define
\begin{align}
\begin{split}
\mathbb{L}(x,y):=&e^{-x}\frac{x^T}{T!}\left(\frac{1}{8}\frac{x}{T+1}+\frac{1}{8 \pip}\right)+\\
&e^{-y}\frac{y^T}{T!}\left(\frac{1}{16}\frac{y}{T+1}+\frac{1}{16 \pip}-\frac{1}{16(\pip)^2}\right).
\label{LB_func}
\end{split}
\end{align}
Thus, From Eq. (\ref{2-1}) and Eq. (\ref{LB_func}), by substituting the lower bounds on $a_1$ and $a_2$ in Eq. (\ref{2-1}), the LHS becomes equal to $\mathbb{L}$ and hence \begin{align} \mathbb{L}\left(\frac{\log(2(\pip)^{T+1})}{\pip-1}, \frac{T}{\pip+1}\right) \leq \mathsf{P}(a_1,a_2), \label{th2_lb}\end{align} is a lower bound on $\mathbb{P}(a_1,...,a_n)$.

\emph{Proof of case 2:} In order to show that the lower bound obtained in the first case still holds if the minimizer value for $a_1$ is greater than $T$, we show that for all values of $a_1>T$ and $T>4$, $\mathsf{P}(a_1,a_2)$ is greater than the lower bound in Eq. \eqref{th2_lb}. To that end, we derive a new lower bound on $\mathsf{P}(a_1,a_2)$, when $a_1$ is greater than $T$, i.e., we find some
$\tilde{P}(T)\leq \mathsf{P}(a_1,a_2)$. Note that the lower bound depends only on $T$. We also find an upper bound on $\mathbb{L}$ in terms of $T$, i.e. $$\widetilde{\mathbb{L}}\left(T\right)\geq \mathbb{L}\left(\frac{\log(2(\pip)^{T+1})}{\pip-1}, \frac{T}{\pip+1}\right).$$
The expression $\tilde{P}(T)$ is an increasing function of $T$, whereas $\widetilde{\mathbb{L}}$ is a decreasing function of $T$. Furthermore, $\tilde{P}(T)=\widetilde{\mathbb{L}}$ at $T^{\ast}=4.87$. Thus, we can conclude that for $T>T^{\ast}$ and for all values of $a_1>T$, $\mathsf{P}(a_1,a_2)$ is greater than lower bound in Eq. \eqref{th2_lb}. 
To get a lower bound on $\mathsf{P}(a_1,a_2)$ in Eq. \eqref{error_reduced}, let us consider the third term and neglect the others. Among all values of $a_1>T$, the third term is minimized at $a_1=T$, so we can write:
\begin{equation}
\tilde{P}(T):=\frac{1}{2}P_1 e^{-T}\sum_{y=T+1}^{\infty}\frac{T^y}{y!}< \min_{a_2>0,a_1>T}{\mathsf{P}(a_1,a_2)}.
\label{lb}
\end{equation}
Now we calculate an upper bound on $\mathbb{L}(A,B)$ for $A=\frac{\log\left(2(\pip)^{T+1}\right)}{\pip-1}$ and $B=\frac{T}{\pip+1}$. The following chain of inequalities hold:
\begin{equation}
\begin{split}
\mathbb{L}(A,B)&=e^{-A}\frac{A^T}{T!}\left[\frac{1}{8}\frac{A}{T+1}+\frac{1}{8 \pip}\right]+e^{-B}\frac{B^T}{T!}\left[\frac{1}{16}\frac{B}{T+1}+\frac{1}{16 \pip}-\frac{1}{16(\pip)^2}\right] \\
&\overset{(i)}{<} e^{-T}\dfrac{(T)^T}{T!}\Big[\dfrac{1}{8}\dfrac{A}{T+1}+\dfrac{1}{8\pip}\Big]+e^{-T/2}\dfrac{(T/2)^T}{T!} \Big[\dfrac{1}{16}\dfrac{B}{T+1}+\dfrac{1}{16}\Big(\dfrac{1}{\pip}-\dfrac{1}{\pip^2}\Big)\Big]\\
&\overset{(ii)}{<}e^{-T}\dfrac{(T)^T}{T!}\dfrac{1}{4}+e^{-T/2}\dfrac{(T/2)^T}{T!} \Big[\dfrac{1}{32}+\dfrac{1}{16}\times\dfrac{1}{4}\Big]\\
&= \dfrac{1}{4} e^{-T}\dfrac{(T)^T}{T!}+\dfrac{3}{64}e^{-T/2}\dfrac{(T/2)^T}{T!}
\\:&=\widetilde{\mathbb{L}}\left(T\right), 
\label{ub}
\end{split}
\end{equation}
where inequality ($i$) holds because $e^{-x}x^T/T!$ is an increasing function of $x$ for $x<T$, and the fact that $A<T$ and $B<T/2$; inequality ($ii$) results from the fact that $\Big(1/\pip-1/(\pip)^2\Big)<1/4$, $A<T$ and $B<T/2$.
The lower bound in Eq. \eqref{lb} is an increasing function of $T$ over integer values and the upper bound in Eq. \eqref{ub} is a decreasing one, so it can be concluded that the upper bound is less than the lower bound, for $T>T^\ast \cong 4.87$, where $T^{\ast}$ is the solution for Eq. \eqref{Tstar}.
\begin{equation}
\frac{1}{2}P_1 e^{-T}\sum_{y=T+1}^{\infty}\frac{T^y}{y!}=\dfrac{1}{4} e^{-T}\dfrac{(T)^T}{T!}+\dfrac{3}{64}e^{-T/2}\dfrac{(T/2)^T}{T!}.
\label{Tstar}
\end{equation}
Thus, we have shown that for $T>4$ and for values of $a_1>T$, $\mathsf{P}(a_1,a_2)$ is greater than $\mathbb{L}(A,B)$.

\emph{Proof of case 3:} We show that as $a_2$ converges to infinity, the minimum value for $\mathsf{P}(a_1,a_2)$ in Eq. \eqref{error_reduced} is greater than the lower bound in Eq. \eqref{th2_lb}. Eq. \eqref{error_reduced} is sum of four poisson CDFs. At $a_2 \rightarrow \infty $ the second term equals zero and the last term equals $\frac{1}{2}P_2 = \frac{1}{16}$; the two other terms are functions of $a_1$ and the sum is minimized at $$a_1^\ast=\frac{\log\left(2(\pip)^{T+1}\right)}{\pip-1}.$$ So the minimum value for $\mathsf{P}(a_1,a_2)$ at $a_2 \rightarrow \infty$ equals:
\begin{align}\frac{1}{4}CDF_{Poisson}(T,a_1^\ast)+\frac{1}{8}(1-CDF_{Poisson}(T,a_1^\ast))+\frac{1}{16},\label{case3}\end{align}
where $CDF_{Poisson}(T,a_1^\ast)=\sum_{y=0}^{T}{e^{-a_1^\ast}\frac{{a_1}^y}{y!}}$. From Eq. \eqref{case3} we see that the minimum value for $\mathsf{P}(a_1,a_2)$ is always greater than $\frac{1}{16}$ at  $a_2 \rightarrow \infty$ and from Eq. \eqref{ub} we find out that for $T>2$ the value of $\widetilde{\mathbb{L}}$ is less than $\frac{1}{16}$, so for all values of $T>2$ the lower bound, $\mathbb{L}(A,B)$, still holds even if the minimizer value for $a_2$ converges to $\infty$ and it completes the proof. 
\end{proof}
\section{Does Sending a Constant Value for Bit `0' Improves the Performance? \label{remark}} 
\begin{theorem}
Among all transmitters which send a constant rate for bit `0', the optimal one sends nothing for bit `0'
\end{theorem}
\begin{proof}
Assuming a constant rate of transmission for bit `0', we are interested in minimizing error probability, which can be calculated similar to Eq. (\ref{errorprob2}). We have
\begin{equation} 
\begin{split}
\mathbb{P}_e&(a_0,a_1,a_2,...) =\dfrac{1}{2}P_{e|1}+\dfrac{1}{2}P_{e|0}=\\
&\dfrac{1}{2}\sum_{i=0}^{\infty}p_i\left[e^{-(a_i+\dfrac{\pi_0}{\pi_1}a_{i+1})}\sum_{ y\leq T}\dfrac{(a_i+\dfrac{\pi_0}{\pi_1}a_{i+1})^y}{y!}\right]+ \\ 
& \dfrac{1}{2}\sum_{i=0}^{\infty}p_i\left[e^{-(a_i+\dfrac{\pi_0}{\pi_1}a_{0})}\sum_{y>T}\frac{(a_i+\pip a_0)^y}{y!}\right],
\label{th3,eq1}
\end{split}
\end{equation}
where $a_0$ represents the interference experienced at the receiver if the last transmitted bit is `0'. In order to prove the theorem, we show that for the values of $a_0$ greater than $0$, the necessary KKT conditions are not satisfied.
If we write the KKT conditions, we will have:
\begin{equation}
\begin{split}
P_0&e^{-a_0-\dfrac{\pi_0}{\pi_1}a_{1}} \dfrac{\left(a_0+\dfrac{\pi_0}{\pi_1}a_{1}\right)^T}{T!} -\\
& P_0 (1+\pip) e^{-(1+\pip) a_0} \dfrac{\left((1+\pip )a_{0}\right)^T}{T!}- \\ 
& \sum_{i=1}^{\infty} P_{i} \pip e^{-a_i-\pip a_0}\dfrac{\left(a_i+\pip a_0\right)^T}{T!}=-\mu_0+\lambda P_0.
\label{th3,eq2}
\end{split}
\end{equation}
Moreover,
\begin{equation}
\begin{split}
&P_ie^{-a_i-\dfrac{\pi_0}{\pi_1}a_{i+1}}\dfrac{\left(a_i+\dfrac{\pi_0}{\pi_1}a_{i+1}\right)^T}{T!} + \\
&P_{i-1}\dfrac{\pi_0}{\pi_1}e^{-a_{i-1}-\dfrac{\pi_0}{\pi_1}a_{i}}\dfrac{\left(a_{i-1}+ \pip a_i\right)^T}{T!} -\\
& P_{i}e^{-a_i-\pip a_0}\dfrac{\left(a_i+\pip a_0\right)^T}{T!}=-\mu_i+\lambda P_i \; \qquad \text{for}\; i= 1,2,\cdots .
\label{th3,eq3}
\end{split}
\end{equation}
Using Eq. \eqref{th3,eq3}, we show that $\mu_0$ in Eq. \eqref{th3,eq2} cannot be zero, implying that $a_0$ equals zero in the optimum solution.
Let us take a summation over $i$ in both sides of Eq. (\ref{th3,eq3}). We get
\begin{equation}
\begin{split}
&P_0 \pip e^{-a_0-\pip a_1}\dfrac{\left(a_0+\pip a_1\right)^T}{T!}-\\
&\sum_{i=1}^{\infty} P_i e^{-a_i-\pip a_0}\dfrac{\left(a_i+\pip a_0\right)^T}{T!}=\\
&\lambda \sum_{i=1}^{\infty} P_i -\sum_{i=1}^{\infty}\mu_i - \sum_{i=1}^{\infty} P_i (1+\pip) e^{-a_i-\pip a_{i+1}} \dfrac{\left(a_i+\pip a_{i+1}\right)^T}{T!}.
\label{th3,eq4}
\end{split}
\end{equation}
Since $\sum_{i=1}^{\infty}P_i=1/2$, we know that the LHS of Eq. (\ref{th3,eq4}) is less than $\lambda/2$. Therefore, by comparing Eq. (\ref{th3,eq2}) and Eq. (\ref{th3,eq4}) we can write
\begin{equation}
\begin{split}
&LHS \;of\; Eq.\; (\ref{th3,eq2}) =P_0e^{-a_0-\dfrac{\pi_0}{\pi_1}a_{1}}\dfrac{\left(a_0+\dfrac{\pi_0}{\pi_1}a_{1}\right)^T}{T!}-\\
&\sum_{i=1}^{\infty} P_{i} \pip e^{-a_i-\pip a_0}\dfrac{\left(a_i+\pip a_0\right)^T}{T!}-\\
& P_0(1+\pip)e^{-(1+\pip) a_{0}}\dfrac{\left((1+\pip) a_{0}\right)^T}{T!} \\
&\leq LHS\; of\; Eq.\;(\ref{th3,eq4}) = P_0 \pip e^{-a_0-\pip a_{1}}\dfrac{\left(a_0+\pip a_{1}\right)^T}{T!}-\\
&\sum_{i=1}^{\infty}P_ie^{-a_i-\pip a_0}\dfrac{\left(a_i+\pip a_0\right)^T}{T!} \leq \dfrac{\lambda}{2} = \lambda P_0.
\label{th3,eq5}
\end{split}
\end{equation}
We have shown that the LHS of Eq. (\ref{th3,eq2}) is less than $\lambda P_0$, so we can conclude that $\mu_0 \neq 0$ which means in all points satisfying KKT condition we should have $a_0 =0$ or equivalently we should send nothing for bit `0'.
\end{proof}
\section{Looking Inside Minimum Solution for Eq. (\ref{fixedDistMin})\label{fixeddist}}
We are interested in minimizing Eq. \ref{fixedDistMin} subject to $\mathbb{E}[f(I)]= K$. More specifically, for simplicity, let us assume that $I$ is a discrete R.V., taking values $a_1, a_2, \cdots, a_r \geq 0$ with probabilities $p_1, \cdots, p_r$ where $\sum_jp_j=1$. Let $f(a_j)=\dfrac{b_j}{\pi_0}\geq 0$, then we want to minimize
\begin{align}
\sum_{j=1}^rp_j\left[e^{-a_j-b_j}\sum_{0\leq y<T}\frac{(a_j+b_j)^y}{y!}\right] 
\label{fixed_int_3}
\end{align}
over all $b_j\geq 0$ subject to $\sum_{j}p_j b_j=\pi_0 K$. Since the constraint is linear, regularity conditions for KKT hold and from the KKT equations we can obtain:
\begin{align}
p_je^{-a_j-b_j}\frac{(a_j+b_j)^T}{T!}=-\mu_i+\lambda p_j
\label{fixed_KKT}
\end{align}
where $\mu_ib_i=0$, $\mu_i\geq 0$. The function $e^{-x}x^T/T!$ is increasing in $[0,T]$ and decreasing from $[T,\infty]$. Therefore, $e^{-x}x^T/T!=\lambda$ has two solutions $C_1$ and $C_2$ one of which is less than $T$ and the other greater than $T$. Thus, assuming that $b_i>0$, for some constant $\lambda\geq 0$, one of the following holds: 
\begin{itemize}
\item $b_j=0$ and $e^{-a_j}\frac{a_j^T}{T!}\leq \lambda$.
\item $a_j+b_j=\alpha_1<T$ where $e^{-\alpha_1}\frac{\alpha_1^T}{T!}= \lambda$.
\item $a_j+b_j=\alpha_2>T$ where $e^{-\alpha_2}\frac{\alpha_2^T}{T!}= \lambda$.
\end{itemize}
In this Appendix, we show that the minimization problem described in Eq. (\ref{fixed_int_3}) has two solutions.
First we show that there does not exists more than one value $a_j$ which maps to $b_j= C_1 - a_j$. Assume it is not true and for some $j_1$, $j_2$ we have $a_{j1} + b_{j1}=a_{j2} + b_{j2} = C_1$ where $b_{j2} > 0$. Let us increase $b_{j1}$ by $\epsilon/p_{j1}$ and decrease $b_{j2}$ by $\epsilon/p_{j2}$ in Eq. (\ref{fixed_int_3}), as a result the constraint on the average transmission rate, i.e., $\sum_j p_j \dfrac{b_j}{\pi_0}=K$ is kept constant. The first derivative of function we are minimizing in Eq. (\ref{fixed_int_3}), with respect to $\epsilon$ is
\begin{align}
\begin{split}
-&\left[ e^{-C_1-\epsilon/p_{j1}}\frac{(C_1+\epsilon/p_{j1})^T}{T!}\right]+\\
&\left[ e^{-C_1+\epsilon/p_{j2}}\frac{(C_1-\epsilon/p_{j2})^T}{T!}\right]
\end{split}
\end{align}
which is zero at $\epsilon=0$ (both terms are equal to $\lambda$). The second derivative at zero is
\begin{align}{\frac{1}{p_{j1}}\lambda\frac{C_1-T}{C_1}+\frac{1}{p_{j2}}\lambda\frac{C_1-T}{C_1}}\end{align}
which should be non-negative; but since $C_1$ is less than T, it is negative. Thus, this is a contradiction, indicating that $a_j+b_j=C_1$ does not happen for two $j$’s. So at most one point is mapped to $C_1$.
On the other hand, if $a_{j1} + b_{j1}=a_{j2} + b_{j2} = C_2$, the second derivative is always non-negative at zero and if we form the Hessian matrix it can be shown that it is positive semi-definite with all eigenvalues non-negative; Thus, $\forall j\in[1,2,...,r]:\, a_{j} + b_{j} = C_2$ is a local minimum. That is the transmitting function $b_j=f(a_j)=C_2-a_j$ , the proposed scheme, is locally optimum.
Assume that for some $j_1$ we have $a_{j1} + b_{j1}=C_1$ and for some $j_2$, $a_{j2} + b_{j2} = C_2$. Using the same method as for the previous case, it can be shown that a necessary condition for being a local minimum is
\begin{align}\frac{p_{j2}}{p_{j1}}\leq \frac{C_1}{C_2}\frac{C_2-T}{T-C_1}\end{align}
\\
We have shown that $\frac{C_1}{C_2}\frac{T-C_2}{C_1-T}<1$. Thus,
\[p_{j2}<p_{j1}\]
As has been shown above, only one $j$ can be mapped to $C_1$. So in addition to the case that all $j$s are mapped to $C_2$ (proposed scheme), there may exist other local minimums in which the largest $p_j$ is mapped to $C_1$ , i.e., $f(a_j)= C_1-a_j$, and other $p_j$s are mapped to $C_2$ , i.e., $b_i = C_2 - a_i$, for all $i \neq j$. 
\end{document}